\documentclass[manuscript, table, nonacm]{acmart}

\usepackage{xcolor}
\usepackage{amsmath}
\usepackage{graphicx}
\usepackage{caption}
\usepackage{subcaption}
\usepackage{verbatim}
\usepackage{latexsym}
\usepackage{setspace}
\usepackage{mathtools}
\usepackage{bbm}
\usepackage{subcaption}
\usepackage{algorithm}
\usepackage{algpseudocode}

\newtheorem{definition}{Definition}
\newtheorem{theorem}{Theorem}
\newtheorem{proposition}{Proposition}
\newtheorem{problem*}{Problem}

\newtheorem{corollary}{Corollary}

\newtheorem{remark}{Remark}
\newtheorem{assumption}{Assumption}

\begin{document}
\title{Robust Event-Driven Interactions in Cooperative Multi-Agent Learning}
\author{Daniel Jarne Ornia}
\email{d.jarneornia@tudelft.nl}
\affiliation{%
  \institution{DCSC, Delft University of Technology}
  \streetaddress{Mekelweg 2}
  \city{Delft}
  \state{The Netherlands}
  \postcode{2628CD}
}
\author{Manuel Mazo Jr.}
\email{m.mazo@tudelft.nl}
\affiliation{%
  \institution{DCSC, Delft University of Technology}
  \streetaddress{Mekelweg 2}
  \city{Delft}
  \state{The Netherlands}
  \postcode{2628CD}
}
\thanks{This work was partially supported by the ERC Starting Grant SENTIENT \# 755953}
\begin{abstract}
We present an approach to reduce the communication required between agents in a Multi-Agent Reinforcement Learning system by exploiting the inherent robustness of the underlying Markov Decision Process. We compute so-called robustness surrogate functions (off-line), that give agents a conservative indication of how far their state measurements can deviate before they need to update other agents in the system with new measurements. This results in fully distributed decision functions, enabling agents to decide when it is necessary to communicate state variables. We derive bounds on the optimality of the resulting systems in terms of the discounted sum of rewards obtained, and show these bounds are a function of the design parameters. Additionally, we extend the results for the case where the robustness surrogate functions are learned from data, and present experimental results demonstrating a significant reduction in communication events between agents.
\end{abstract}
\maketitle
\thispagestyle{empty}
\section{Introduction}
In the last two decades we have seen a surge of learning-based techniques applied to the field of multi agent game theory, enabling the solution of larger and more complex problems, both model based and model free \cite{hu1998multiagent,busoniu2008comprehensive,nowe2012game}. Lately, with the wide adoption of Deep Learning techniques for compact representations of value functions and policies in model-free problems \cite{mnih2013playing,lillicrap2015continuous,van2016deep}, the field of Multi-Agent Reinforcement Learning (MARL) has seen an explosion in the applications of such algorithms to solve real-world problems \cite{lowe2017multi}. However, this has naturally led to a trend where both the amount of data handled in such data driven approaches and the complexity of the targeted problems grow exponentially. In a MARL setting where communication between agents is required, this may inevitably lead to restrictive requirements in the frequency and reliability of the communication to and from each agents (as it was already pointed out in \cite{panait2005cooperative}).

The effect of asynchronous communication in dynamic programming problems was studied already in \cite{bertsekas1982distributed}. In particular, one of the first examples of how communication affects learning and policy performance in MARL is found in \cite{tan1993multi}, where the author investigates the impact of agents sharing different combinations of state variable subsets or Q values. After that, there have been multiple examples of work studying different types of communication in MARL and what problems arise from it \cite{sen1994learning,ackley1994altruism,szer2004improving,kok2004sparse}. In this line, in \cite{zhang2013coordinating} actor coordination minimization is addressed and in \cite{foerster2016learning,kim2019learning} authors allow agents to choose a communication action and receive a reward when this improves the policies of other agents. In \cite{chen2018communication} multi agent policy gradient methods are proposed with convergence guarantees where agents communicate gradients based on some trigger conditions, and in \cite{lin2019communication} agents are allowed to communicate a simplified form of the parameters that determine their value function.

We focus particularly in a \emph{centralised training - decentralised execution}, where agents must communicate state measurements to other agents in order to execute the distributed policies. Such a problem represents most real applications of MARL systems: It is convenient to train such systems in a simulator, in order to centrally learn all agents' value functions and policies. But if the policies are to be executed in a live (real) setting, agents will have access to different sets of state variables that need to be communicated with each-other. In this case, having non-reliable communication leads to severe disruptions in the robustness of the distributed policies' performance. The authors in \cite{lin2020robustness} demonstrated experimentally how very small adversarial disruptions in state variable communications leads to a collapse of the performance of general collaborative MARL systems. In this regard, \cite{da2020uncertainty} proposes learning an ``adviser'' model to fall back on when agents have too much uncertainty in their state measurements, and more recently in \cite{karabag2022planning} the authors enable agents to run simulated copies of the environment to compensate for a disruption in the communication of state variables, and in \cite{xue2021mis} agents are trained using adversarial algorithms to achieve more robust policies. This lack of robustness in communicative multi-agent learning presents difficulties when trying to design efficient systems where the goal is to communicate less often.

With this goal in mind, we can look into event triggered control (ETC) as a strategy to reduce communication \cite{tabuada2007event, mazo2008event} in a networked system by trading off communication for robustness against state measurement deviations. This has been applied before in linear multi-agent systems \cite{6068223} and non-linear systems \cite{9683215,zhong2014event}. In \cite{vamvoudakis2018model} and \cite{sahoo2015neural} ideas on how to use ETC on model-free linear and non-linear systems were explored. Additionally in other learning problems such as \cite{solowjow2020event}, where authors show how event triggered rules can be applied to learn model parameters more efficiently, and in \cite{george2020distributed} by applying a similar principle to demonstrate how ETC can be used to compute stochastic gradient descent steps in a decentralised event-based manner.

\subsection{Main Contribution}
We consider in this work a general cooperative MARL scenario where agents have learned distributed policies that must be executed in an on-line scenario, and that depend on other agent's measurements. We propose a constructive approach to synthesise communication strategies that minimise the amount of communication required and guarantee a minimum performance of the MARL system in terms of cumulative reward when compared to an optimal policy. We construct so-called \emph{robustness surrogate} functions, which quantify the robustness of the system to disturbances in agent state variable measurements, allowing for less communication in more robust state regions. Additionally, we consider the case where these surrogate functions are learned through the \emph{scenario approach} \cite{campi2020scenario,calafiore2006scenario}, and show how the guarantees are adapted for learned approximated functions.
\section{Preliminaries}
\subsection{Notation}
We use calligraphic letters for sets and regular letters for functions $f:\mathbb{R}^m\to\mathbb{R}^n$.  We say a function $f: \mathbb{R}_+\to \mathbb{R}_+$ is $f\in\mathcal{K}_\infty$ if it is continuous, monotonically increasing and $f(0)=0,\,\,\lim_{a\to\infty}f(a) = \infty$.
We use $\mathcal{F}$ as the $\sigma$-algebra of events in a probability space, and $P$ as a probability measure $P:\mathcal{F}\to [0,1]$. We use $E[\cdot]$ and $\operatorname{Var}[\cdot]$ for the expected value and the variance of a random variable. 
We use $\|\cdot\|_\infty$ as the sup-norm, $|\cdot |$ as the absolute value or the cardinality, and $\langle v,u\rangle$ as the inner product between two vectors. We say a random process $X_n$ converges to a random variable $X$ \emph{almost surely} (a.s.) as $t\to\infty$ if it does so with probability one for any event $\omega\in\mathcal{F}$. For a conditional expectation, we write $E[X|Y]\equiv E_{Y}[X]$.

\subsection{MDPs and Multi-Agent MDPs}
We first present the single agent MDP formulation. 
\begin{definition}\label{def:MDP}[Markov Decision Process]
A Markov Decision Process (MDP) is a tuple $(\mathcal{X},\mathcal{U},P,r)$ where $\mathcal{X}$ is a set of states, $\mathcal{U}$ is a set of actions, $P: \mathcal{U}\to \mathbb{P}^{|\mathcal{X}|\times|\mathcal{X}|}$ is a probability measure of the transitions between states and $r:\mathcal{X}\times \mathcal{U}\times\mathcal{X}\to\mathbb{R}$ is a reward function, such that $r(\mathbf{x},u,\mathbf{x}')$ is the reward obtained when action $u$ is performed at state $\mathbf{x}$ and leads to state $\mathbf{x}'$. 
\end{definition}
In general, an agent has a policy $\pi:\mathcal{X}\to\mathbb{P}({\mathcal{U}})$, that maps the states to a probability vector determining the chance of executing each action. We can extend the MDP framework to the case where multiple agents take simultaneous actions on an MDP. For the state transition probabilities we write in-distinctively $P_{\mathbf{x}\mathbf{x}'}(u)\equiv P(\mathbf{x},\mathbf{x}',u)$, and for the reward obtained in two consecutive states $\mathbf{x}_t,\mathbf{x}_{t+1}$ we will write $r_t\equiv r(\mathbf{x}_t,u_t,\mathbf{x}_{t+1})$.
\begin{definition}\label{def:MMDP}[Collaborative Multi-Agent MDP]
A Collaborative Multi-Agent Markov Decision Process (c-MMDP) is a tuple $(\mathcal{N},\mathcal{X},\mathcal{U}^n,P,r)$ where $\mathcal{N}$ is a set of $n$ agents, $\mathcal{X}$ is a cartesian product of \emph{metric} state spaces $\mathcal{X} = \prod_{i\in\mathcal{N}}\mathcal{X}_i$, $\mathcal{U}^n=\prod_{i\in\mathcal{N}}\mathcal{U}_i$ is a joint set of actions, $P: \mathcal{U}^n\to \mathbb{P}^{\mathcal{X}\times\mathcal{X}}$ is a probability measure of the transitions between states and $r:\mathcal{X}\times \mathcal{U}^n\times\mathcal{X}\to\mathbb{R}$ is a reward function. 
\end{definition}
\begin{assumption}\label{as:1}
We assume that each agent $i$ has access to a set $\mathcal{X}_{i}\subset \mathcal{X}$, such that the observed state for agent $i$ is  $\mathbf{x}(i)\in\mathcal{X}_{i}$. That is, the global state at time $t$ is $\mathbf{x}_t=(\mathbf{x}_t({1})\,\,\mathbf{x}_t({2})\,\,...\,\mathbf{x}_t({n}))^T$. Furthermore, we assume that the space $\mathcal{X}$ accepts a sup-norm $\|\cdot\|_{\infty}$.
\end{assumption}
In the c-MMDP case, we can use $U\in \mathcal{U}^n$ to represent a specific joint action $U:=\{U(1), U(2),..., U(n)\}$, and $\Pi := \{\pi_{1},\pi_{2},...,\pi_{n}\}$ to represent the joint policies of all agents such that $\Pi:\mathcal{X}\to\mathcal{U}^n$ . We assume in this work that agents have a common reward function, determined by the joint action. That is, even if agents do not have knowledge of the actions performed by others, the reward they observe still depends on everyone's joint action. Additionally, we assume in the c-MMDP framework that the control of the agents is fully distributed, with each agent having its own 
(deterministic) policy $\pi_{i}$ that maps the global state to the individual action, i.e. $\pi_{i}:\mathcal{X}\to\mathcal{U}_i$.
We define the optimal policy in an MDP as the policy $\pi^*$ that maximises the expected discounted reward sum $E[\sum_{t=1}^{\infty}\gamma^t r_t\, | \pi,\mathbf{x}_0]$ $\forall \mathbf{x}_0 \in \mathcal{X}$ over an infinite horizon, for a given discount $\gamma\in (0,1)$. 
The optimal joint (or \emph{centralised}) policy in a c-MMDP is the joint policy $\Pi^*$ that maximises the discounted reward sum in the ``centrally controlled''  MDP, and this policy can be decomposed in a set of agent-specific optimal policies $\Pi^* = \{\pi^{*}_1,\pi^{*}_2,...,\pi^{*}_n\}$.
\begin{remark}\label{rem:1}
Assumption \ref{as:1} is satisfied in most MARL problems where the underlying MDP represents some real physical system (\emph{e.g.} robots interacting in a space, autonomous vehicles sharing roads, dynamical systems where the state variables are metric...). In the case where $\mathcal{X}$ is an abstract discrete set, we can still assign a trivial bijective map $I:\mathcal{X}\to \mathbb{N}$ and compute distances on the mapped states $\|\mathbf{x}_1-\mathbf{x}_2\|_{\infty}\equiv\|I(\mathbf{x}_1)-I(\mathbf{x}_2)\|_{\infty}$. However, we may expect the methods proposed in this work to have worse results when the states are artificially numbered, since the map $I$ may have no relation with the transition probabilities (we come back to this further in the work).
\end{remark}
\subsection{Value Functions and Q-Learning in c-MMDPs}
Consider a c-MMDP, and let a value function under a joint policy $\Pi$, $V^{\Pi}:\mathcal{X}\to\mathbb{R}$ be $V^{\Pi}(\mathbf{x}) = \sum_{\mathbf{x}'}P_{\mathbf{x}\mathbf{x}'}(\Pi(\mathbf{x}))(r(\mathbf{\mathbf{x}},\Pi(\mathbf{x}),\mathbf{x}')+\gamma V^{\Pi}(\mathbf{x}'))$. There exists an optimal value function $V^*$ for a centralised controller in a c-MMDP that solves the Bellman equation:
\begin{equation*}
V^*(\mathbf{x}) := \max_{U}\sum_{\mathbf{x}'}P_{\mathbf{x}\mathbf{x}'}(U)(r(\mathbf{x},U,\mathbf{x}') + \gamma V^*(\mathbf{x}')).
\end{equation*}
Now consider so-called Q functions $Q:\mathcal{X}\times\mathcal{U}^n\to\mathbb{R}$ on the centrally controlled c-MMDP \cite{watkins1992q}, such that the optimal Q function satisfies
\begin{equation*}
Q^*(\mathbf{x},U) :=\sum_{\mathbf{x}'}P_{\mathbf{x}\mathbf{x}'}(U)(r(\mathbf{x},U,\mathbf{x}') + \gamma \max_{U'}Q^*(\mathbf{x}',U')),
\end{equation*}
and the optimal centralised policy is given by $U^* :=\pi^*(\mathbf{x}) = \operatorname{argmax}_U Q^*(\mathbf{x},U)$. 
Additionally, $ \max_{U}Q^*(\mathbf{x},U) = V^*(\mathbf{x}) =E[\sum_{t=1}^{\infty}\gamma^t r(\mathbf{x},\Pi^*(\mathbf{x}),\mathbf{x}')\, |\mathbf{x}_0]$. 

\section{Information sharing between collaborative agents: Problem Formulation}
Consider now the case of a c-MMDP where each agent has learned a distributed policy $\pi_{i}:\mathcal{X}\to\mathcal{U}_i$. We are interested in the scenario where the state variable $\mathbf{x}_t\in\mathcal{X}$ at time $t$ is composed by a set of joint observations from all agents, and these observations need to be communicated to other agents for them to compute their policies. 
\begin{assumption}
Agents have a set of optimal policies $\Pi^*$ available for executing on-line and a global optimal function $Q^*$ (learned as a result of e.g. a multi-agent actor critic algorithm \cite{lowe2017multi}).
\end{assumption}
Consider the case where at a given time $t$, a subset of agents $\hat{\mathcal{N}}_t\subseteq \mathcal{N}$ does not share their state measurements with other agents. Let $t_i$ be the last time agent $i$ transmitted its measurement. We define $\hat{\mathbf{x}}_t\in\mathcal{X}$ as 
\begin{equation}
\hat{\mathbf{x}}_t := \left(\mathbf{x}_{t_1}(1),\mathbf{x}_{t_2}(2),...,\mathbf{x}_{t_n}(n)\right).
\end{equation}
That is, $\hat{\mathbf{x}}_t$ is the last known state corresponding to the collection of agent states last transmitted, at time $t$.
Then, the problem considered in this work is as follows.
\begin{problem*}
Consider a c-MMDP with a set of optimal shared state policies $\Pi^*$. Synthesise strategies that minimise the communication events between agents and construct distributed policies $\hat{\Pi}$ that keep the expected reward within some bounds of the optimal rewards, these bounds being a function of design parameters.
\end{problem*}

\section{Efficient Communication Strategies}
To solve the problem of minimizing communication, we can first consider a scenario where agents can request state measurements from other agents. Consider a c-MMDP where agents have optimal policies $\Pi^*$. If agents are allowed to request state observations from other agents at their discretion, a possible algorithm to reduce the communication when agents execute their optimal policies is to use sets of neighbouring states $\mathcal{D}:\mathcal{X}\to 2^{\mathcal{X}}$ such that $\mathcal{D}(\mathbf{x})=\{\mathbf{x}':\|\mathbf{x}-\mathbf{x}'\|\leq d\}$ for some maximum distance $d$. Agents could compute such sets for each point in space, and request information from others only if the optimal action changes for any state $\mathbf{x}'\in\mathcal{D}(\mathbf{x})$. 
This approach, however, is not practical on a large scale multi agent system. First, it requires agents to request information, which could already be considered a communication event and therefore not always desirable. Additionally, computing the sets ``on the fly'' has a complexity of $\mathcal{O}(|\mathcal{X}|^2)$ in the worst case, and it has to be executed at every time-step by all agents. We therefore propose an approach to reduce communication in a MARL system where agents do not need to request information, but instead send messages (or not) to other agents based on some triggering rule.
\subsection{Event-Driven Interactions}
To construct an efficient communication strategy based on a distributed triggering approach, let us first define a few useful concepts.
In order to allow agents to decide when is it necessary to transmit their own state measurements, we define the \emph{robustness indicator} $\Gamma:\mathcal{X}\to \mathbb{R}_{\geq 0}$ as follows.
\begin{definition}\label{def:gam}
For a c-MMDP with optimal global $Q^*:\mathcal{X}\times\mathcal{U}^n\to\mathbb{R}$, we define the \emph{robustness surrogate} $\Gamma_\alpha:\mathcal{X}\to \mathbb{R}_{\geq 0}$ with sensitivity parameter $\alpha\in\mathbb{R}_{\geq 0}$ as:
\begin{equation*}\begin{aligned}
\Gamma_\alpha (\mathbf{x}):=\max\{d\,\lvert \, & \forall \mathbf{x}':\|\mathbf{x}'-\mathbf{x}\|_{\infty}\leq d\Rightarrow\\
\Rightarrow &Q^*(\mathbf{x}',\Pi^*(\mathbf{x}))\geq V^*(\mathbf{x}')-\alpha\}.
\end{aligned}
\end{equation*}
\end{definition}
The function $\Gamma_\alpha$ gives a maximum distance (in the sup-norm) such that for any state $\mathbf{x}'$ which is $\Gamma_\alpha$ close to $\mathbf{x}$ guarantees the action $\Pi^*(\mathbf{x})$ has a $Q$ value which is $\alpha$ close to the optimal value in $\mathbf{x}'$. A representation can be seen in Figure \ref{fig:gam}. 
\begin{figure}[t!]
     \centering
         \includegraphics[width=0.5\linewidth]{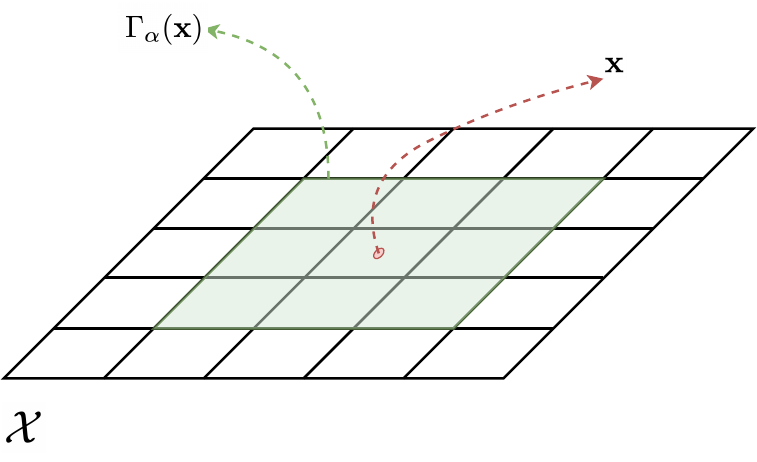}
  \caption{Robustness surrogate representation.}
        \label{fig:gam}
\end{figure}
Computing the function $\Gamma_\alpha$ in practice may be challenging, and we cover this in detail in following sections.
\begin{algorithm}
\caption{Self-Triggered state sharing}\label{al:2}
\begin{algorithmic}
\State Initialise $\mathcal{N}$ agents at $\mathbf{x}_0$;
\State Initialise last-known state vector $\hat{\mathbf{x}}_0 = \mathbf{x}_0 ,\,\,\,i\in\mathcal{N}$
\State $t=0$, 
\While{$t<t_{max}$}
\For{$i\in\mathcal{N}$}
\If{$\|\mathbf{x}_t({i})-\hat{\mathbf{x}}_{t-1}(i)\|_{\infty}>\Gamma_\alpha(\hat{\mathbf{x}}_{t-1})$}
\State $\hat{\mathbf{x}}_{t}(i)\leftarrow \mathbf{x}_t({i)}$
\State Send updated $\hat{\mathbf{x}}_{t}(i)$ to all $\mathcal{N}_{-i}$;
\EndIf
\State Execute action $\hat{U}^*_i=\pi_i^*(\hat{\mathbf{x}})$;
\EndFor
\State $t++$;
\EndWhile
\end{algorithmic}
\end{algorithm}
\begin{proposition}\label{prop:2}
Consider a c-MMDP communicating and acting according to Algorithm \ref{al:2}. Let $\hat{\mathbf{x}}^i_t$ be the last known joint state stored by agent $i$ at time $t$, and $\mathbf{x}_t$ be the true state at time $t$. Then, it holds:
\begin{equation}\label{eq:prop1}
\hat{\mathbf{x}}^i_t=\hat{\mathbf{x}}_t\,\,\,\forall i\in\mathcal{N},
\end{equation}
\begin{equation}\label{eq:prop2}
\|\hat{\mathbf{x}}_t-\mathbf{x}_t\|_{\infty}\leq \Gamma_{\alpha}(\hat{\mathbf{x}}_t)\,\, \forall t.
\end{equation}
\end{proposition}
\begin{proof}
Properties \eqref{eq:prop1} and \eqref{eq:prop2} hold by construction. First, all agents update their own $\hat{\mathbf{x}}^i_t$ based on the received communication, therefore all have the same last-known state. Second, whenever the condition $\|\mathbf{x}_t({i})-\hat{\mathbf{x}}_{t-1}(i)\|_{\infty}>\Gamma_\alpha(\hat{\mathbf{x}})$ is violated, agent $i$ transmits the new state measurement to others, and $\hat{\mathbf{x}}_t$ is updated. Therefore $\|\mathbf{x}_t-\hat{\mathbf{x}}_t\|_{\infty}>\Gamma_\alpha(\hat{\mathbf{x}}_t)$ holds for all times.
\qed
\end{proof}
Now let us use $\hat{r}_t=r(\mathbf{x}_t,\Pi^*(\hat{\mathbf{x}}_t),\mathbf{x}_{t+1})$ as the reward obtained when using the delayed state $\hat{\mathbf{x}}_t$ as input for the optimal policies. We then present the following result.
\begin{theorem}\label{the:1}
Consider a c-MMDP and let agents apply Algorithm \ref{al:2} to update the delayed state vector $\hat{\mathbf{x}}_t$. Then it holds $\forall \mathbf{x}_0\in\mathcal{X}$:
\begin{equation*}
E_{\mathbf{x}_0}[\sum_{t=0}^{\infty}\gamma^{t}r(\mathbf{x}_t,\Pi^*(\hat{\mathbf{x}}_t),\mathbf{x}_{t+1})]\geq  V^*(\mathbf{x}_0)-\alpha \frac{\gamma}{1-\gamma}.
\end{equation*}
\end{theorem}
\begin{proof}
From Proposition \ref{prop:2}, $\|\mathbf{x}_t-\hat{\mathbf{x}}_t\|_{\infty}\leq \Gamma_\alpha(\hat{\mathbf{x}}_t)\,\forall t$, and
 recalling the expression for the optimal $Q$ values:
\begin{equation}\label{eq:lem22}
Q^*(\mathbf{x}_t,\Pi^*(\hat{\mathbf{x}}_t))=E_{\mathbf{x}_t}[\hat{r}_t + \gamma V^*(\mathbf{x}_{t+1})]\geq V^*(\mathbf{x}_t)-\alpha .
\end{equation}
Now let $\hat{V}(\mathbf{x}_0):=E_{\mathbf{x}_0}[\sum_{t=0}^{\infty}\gamma^{t}\hat{r}_t]$ be the value of the policy obtained from executing the actions $\Pi^*(\hat{\mathbf{x}}_t)$. Then:
\begin{equation}\label{eq:lem23}\begin{aligned}
E_{\mathbf{x}_0}&[\sum_{t=0}^{\infty}\gamma^{t}\hat{r}_t] =E_{\mathbf{x}_0}[\hat{r}_0 + \gamma V^{\hat{\Pi}}(\mathbf{x}_1)]=\\
=&E_{\mathbf{x}_0}[\hat{r}_0 + \gamma \hat{V}(\mathbf{x}_1)+\gamma V^*(\mathbf{x}_1)-\gamma V^*(\mathbf{x}_1) ]=\\
=&E_{\mathbf{x}_0}[\hat{r}_0 + \gamma V^*(\mathbf{x}_1)]+\gamma E_{\mathbf{x}_0}[\hat{V}(\mathbf{x}_1)-V^*(\mathbf{x}_1)].
\end{aligned}
\end{equation}
Then, substituting \eqref{eq:lem22} in \eqref{eq:lem23}:
\begin{equation}\label{eq:lem24}\begin{aligned}
E_{\mathbf{x}_0}[\sum_{t=0}^{\infty}\gamma^{t}\hat{r}_t]\geq &V^*(\mathbf{x}_0)-\alpha +\gamma E_{\mathbf{x}_0}[\hat{V}(\mathbf{x}_1)-V^*(\mathbf{x}_1)].
\end{aligned}
\end{equation}
Now, observe we can apply the same principle as in \eqref{eq:lem23} for the last term in \eqref{eq:lem24},
\begin{equation}\label{eq:lem25}\begin{aligned}
&\hat{V}(\mathbf{x}_1)-V^*(\mathbf{x}_1)=E_{\mathbf{x}_1}[\hat{r}_1+\gamma \hat{V}(\mathbf{x}_2)]-V^*(\mathbf{x}_1)=\\
&=Q^*(\mathbf{x}_1,\Pi^*(\hat{\mathbf{x}}_1))+\gamma E_{\mathbf{x}_1}[\hat{V}(\mathbf{x}_2)-V^*(\mathbf{x}_2)]-V^*(\mathbf{x}_1)\geq\\
&\geq V^*(\mathbf{x}_1)-\alpha -V^*(\mathbf{x}_1)+\gamma E_{\mathbf{x}_1}[\hat{V}(\mathbf{x}_2)-V^*(\mathbf{x}_2)]=\\
&=-\alpha +\gamma E_{\mathbf{x}_1}[\hat{V}(\mathbf{x}_2)-V^*(\mathbf{x}_2)].
\end{aligned}
\end{equation}
Substituting \eqref{eq:lem25} in \eqref{eq:lem24}:
\begin{equation}\label{eq:lem26}\begin{aligned}
E_{\mathbf{x}_0}[\sum_{t=0}^{\infty}\gamma^{t}\hat{r}_t] \geq V^*(\mathbf{x}_0)-\alpha -\gamma\alpha +\gamma^2 E_{\mathbf{x}_0}[E_{\mathbf{x}_1}[\hat{V}(\mathbf{x}_2)-V^*(\mathbf{x}_2)]].
\end{aligned}
\end{equation}
Now it is clear that, applying \eqref{eq:lem25} recursively:
\begin{equation}\label{eq:lem27}\begin{aligned}
E_{\mathbf{x}_0}&[\sum_{t=0}^{\infty}\gamma^{t}\hat{r}_t] \geq V^*(\mathbf{x}_0)-\alpha -\\
&-\gamma\alpha +\gamma^2 E_{\mathbf{x}_0}[E_{\mathbf{x}_1}[\hat{V}(\mathbf{x}_2)-V^*(\mathbf{x}_2)]]\geq V^*(\mathbf{x}_0)-\alpha \sum_{k=0}^{\infty}\gamma^{k}.
\end{aligned}
\end{equation}
Substituting $\sum_{k=0}^{\infty}\gamma^{k}=\frac{\gamma}{1-\gamma}$ in \eqref{eq:lem27}:
\begin{equation*}\begin{aligned}
E_{\mathbf{x}_0}&[\sum_{t=0}^{\infty}\gamma^{t}\hat{r}_t] \geq V^*(\mathbf{x}_0)-\alpha \frac{\gamma}{1-\gamma}.
\end{aligned}
\end{equation*}
\qed
\end{proof}
\section{Robustness Surrogate and its Computation}\label{sec:data}
The computation of the robustness surrogate $\Gamma_\alpha$ may not be straight forward. When the state-space of the c-MMDP is metric, we can construct sets of neighbouring states for a given $\mathbf{x}$. Algorithm \ref{al:data} produces an exact computation of the robustness surrogate $\Gamma_\alpha$ for a given c-MMDP and point $\mathbf{x}$.
\begin{algorithm}
\caption{Computation of Robustness Indicator}\label{al:data}
\begin{algorithmic}
\State Initialise $\mathbf{x}$.
\State Initialise $d=1$.
\State Done = \emph{False}
\While{Not Done}
\State Compute Set $\mathcal{X}^d:=\{\mathbf{x}':\|\mathbf{x}-\mathbf{x}'\|= d\}$;
\If{$\exists \mathbf{x}'\in\mathcal{X}^d:Q^*(\mathbf{x}',\Pi^*(\mathbf{x}))\leq V^*(\mathbf{x}')-\alpha$}
\State Done = \emph{True}
\Else{}
$d++$
\EndIf
\EndWhile
\State $\Gamma_\alpha(\mathbf{x})=d-1$
\end{algorithmic}
\end{algorithm}
Observe, in the worst case, Algorithm \ref{al:data} has a complexity of $O(|\mathcal{X}|)$ to compute the function $\Gamma_\alpha(\mathbf{x})$ for a single point $\mathbf{x}$. If this needs to be computed across the entire state-space, it explodes to an operation of worst case complexity $O(|\mathcal{X}|^2)$. In order to compute such functions more efficiently while retaining probabilistic guarantees, we can make use of the Scenario Approach for regression problems \cite{campi2020scenario}. 
\subsection{Learning the Robustness Surrogate with the Scenario Approach}\label{sec:SVR}
The data driven computation of the function $\Gamma_\alpha$ can be proposed in the terms of the following optimization program. Assume we only have access to a uniformly sampled set $\mathcal{X}_S\subset \mathcal{X}$ of size $|\mathcal{X}_S|=S$. Let $\hat{\Gamma}_\alpha^{\theta}$ be an approximation of the real robustness surrogate parametrised by $\theta$. To apply the scenario approach optimization, we need $\hat{\Gamma}_\alpha^{\theta}$ to be convex with respect to $\theta$. For this we can use a Support Vector Regression (SVR) model, and embed the state vector in a higher dimensional space trough a feature non-linear map $\phi(\cdot)$ such that $\phi(\mathbf{x})$ is a feature vector, and we use the kernel $k(\mathbf{x}_1,\mathbf{x}_2)=\langle \phi(\mathbf{x}_1),\phi(\mathbf{x}_2)\rangle$. Let us consider sampled pairs $\{(\mathbf{x}_s,y_s)\}_S$, with $y_s=\Gamma_\alpha(\mathbf{x}_s)$ computed through Algorithm \ref{al:data}. 
Then, we propose solving the following optimization problem with parameters $\tau,\rho>0$:
\begin{equation}\begin{aligned}\label{eq:optimization}
\min_{\theta\in\mathcal{X},\kappa\geq 0,b\in\mathbb{R},\atop
\xi_i\geq 0, i=1,2,...,S}& \left(\kappa+\tau\|\theta\|^{2}\right)+\rho \sum_{i=1}^{S} \xi_{i},\\
s.t.& \quad \left|y_{i}-k( \theta, \mathbf{x}_{i})-b\right|-\kappa \leq \xi_{i}, \quad i=1, \ldots, S.
\end{aligned}
\end{equation}
The solution to the optimization problem \eqref{eq:optimization} yields a trade-off between how many points are outside the \emph{prediction tube} $\left|y-k( \theta^*, \mathbf{x}_{i})-b^{*}\right|<\kappa^{*} $ and how large the tube is (the value of $\kappa^*$). Additionally, the parameter $\rho$ enables us to tune how much we want to penalise sample points being out of the prediction tube. Now take $(\theta^*,\kappa^*,b^*,\xi_i^*)$ as the solution to the optimization problem \eqref{eq:optimization}. Then, the learned robustness surrogate function will be:
\begin{equation*}
\hat{\Gamma}_\alpha^{\theta^*} := k( \theta^*, \mathbf{x}_{i})+b^{*}.
\end{equation*}
From Theorem 3 \cite{campi2020scenario}, it then holds for a sample of points $\mathcal{X}_S$ and a number of outliers $s^*:=|\{(\mathbf{x}, y)\in \mathcal{X}_S:\left|y-k( \theta^*, \mathbf{x})-b^{*}\right|>\kappa^{*} \}|$:
\begin{equation}
\begin{aligned}
\Pr^{S}\left\{\underline{\epsilon}\left(s^{*}\right) \leq \Pr\left\{\mathbf{x}:\left|\Gamma_\alpha(\mathbf{x}) -\hat{\Gamma}_\alpha^{\theta^*}(\mathbf{x}) \right|>\kappa^{*}\right\} \leq \bar{\epsilon}\left(s^{*}\right)\right\} \geq 1-\beta
\end{aligned}
\end{equation}
where $\underline{\epsilon}(s^*):=\max \{0,1-\bar{t}(s^*)\},\, \bar{\epsilon}(s^*):=1-\underline{t}(s^*)$, and $\bar{t}(s^*),\underline{t}(s^*)$ are the solutions to the polynomial
\begin{equation*}\begin{aligned}
\binom{S}{s^*} t^{S-s^*}-\frac{\beta}{2 S} \sum_{i=k}^{S-1}\binom{i}{s^*} t^{i-k}
-\frac{\beta}{6 S} \sum_{i=S+1}^{4 S}\binom{i}{s^*}t^{i-s^*}=0.
\end{aligned}
\end{equation*}

Now observe, in our case we would like $\Gamma_\alpha(\mathbf{x}) \geq \hat{\Gamma}_\alpha^{\theta^*}(\mathbf{x})$ to make sure we are never over-estimating the robustness values. Then, with probability larger than $1-\beta$:
\begin{equation}\label{eq:minprob}\begin{aligned}
\bar{\epsilon}\left(s^{*}\right)\geq& \Pr\left\{\mathbf{x}:\left|\Gamma_\alpha(\mathbf{x}) -\hat{\Gamma}_\alpha^{\theta^*}(\mathbf{x}) \right|>\kappa^{*}\right\}\geq \Pr\left\{\mathbf{x}: \Gamma_\alpha(\mathbf{x}) -\hat{\Gamma}_\alpha^{\theta^*}(\mathbf{x}) <- \kappa^{*}\right\} =\\
=& \Pr\left\{\mathbf{x}: \Gamma_\alpha(\mathbf{x}) < \hat{\Gamma}_\alpha^{\theta^*}(\mathbf{x})-\kappa^{*}\right\}.
\end{aligned}
\end{equation}
Therefore, taking $\|\mathbf{x}_t({i})-\hat{\mathbf{x}}_{t-1}(i)\|_{\infty}> \hat{\Gamma}_\alpha^{\theta^*}(\hat{\mathbf{x}}_{t-1})-\kappa^{*}$ as the condition to transmit state measurements for each agent, we know that the probability of using an over-estimation of the true value $\Gamma_\alpha(\mathbf{x}_t)$ is at most $\bar{\epsilon}\left(s^{*}\right)$ with confidence $1-\beta$.

Then, let $\{\hat{U}_t\}$ be the sequence of joint actions taken by the system. The probability of $U_t$ violating the condition $Q^*(\mathbf{x}_t,U_t)\geq V^*(\mathbf{x}_t)-\alpha$ for any $\mathbf{x}_t\in\mathcal{X}$ is at most $\bar{\epsilon}\left(s^{*}\right)$. Then, we can extend the results from Theorem \ref{the:1} for the case where we use a SVR approximation as a robustness surrogate. Define the worst possible suboptimality gap $\iota := \max_{\mathbf{x},U} |V^*(\mathbf{x})-Q^*(\mathbf{x},U)|$. 
\begin{corollary}\label{cor:1}
Let $\hat{\Gamma}_\alpha^{\theta^*}$ obtained from \eqref{eq:optimization} from collection of samples $\mathcal{X}_S$. Then, a c-MMDP communicating according to Algorithm \ref{al:2} using as trigger condition $\|\mathbf{x}_t({i})-\hat{\mathbf{x}}_{t-1}(i)\|_{\infty}> \hat{\Gamma}_\alpha^{\theta^*}(\hat{\mathbf{x}}_{t-1})-\kappa^{*}$ yields, with probability higher than $1-\beta$:
\begin{equation*}
E_{\mathbf{x}_0}[\sum_{t=0}^{\infty}\gamma^{t}\hat{r}_t]\geq V^*(\mathbf{x}_0)-\delta ,
\end{equation*}
with $\delta :=(\alpha+\bar{\epsilon}\left(s^{*}\right)(\iota-\alpha))\frac{\gamma}{1-\gamma}$.
\end{corollary}
\begin{proof}
Take expression \eqref{eq:minprob}, and consider the action sequence executed by the c-MMDP to be $\{\hat{U}_t\}_{t=0}^{\infty}$. We can bound the total expectation of the sum of rewards by considering $\{\hat{U}_t\}_{t=0}^{\infty}$ to be a sequence of random variables that produce $Q^*(\mathbf{x}_t,\hat{U}_t)\geq V^*(\mathbf{x}_t)-\alpha$ with probability $1-\bar{\epsilon}\left(s^{*}\right)$, and $Q^*(\mathbf{x}_t,\hat{U}_t)\geq V^*(\mathbf{x}_t)-\iota$ with probability $\bar{\epsilon}\left(s^{*}\right)$. Then,
\begin{equation}\label{eq:the2}\begin{aligned}
E&_{\mathbf{x}_0}[\sum_{t=0}^{\infty}\gamma^{t}\hat{r}_t]=E[E_{\mathbf{x}_0}[\sum_{t=0}^{\infty}\gamma^{t}\hat{r}_t|\{\hat{U}_t\}]]= E[E_{\mathbf{x}_0}[\hat{r}_0 + \gamma \hat{V}(\mathbf{x}_1)|\{\hat{U}_t\}]]=\\
=& E[E_{\mathbf{x}_0}[\hat{r}_0 + \gamma V^*(\mathbf{x}_1)|\{\hat{U}_t\}]+\gamma E_{\mathbf{x}_0}[\hat{V}(\mathbf{x}_1)-V^*(\mathbf{x}_1)|\{\hat{U}_t\}]].
\end{aligned}
\end{equation}
Observe now, for the first term in \eqref{eq:the2}:
\begin{equation}\label{eq:the21}\begin{aligned}
E[&E_{\mathbf{x}_0}[\hat{r}_0 + \gamma V^*(\mathbf{x}_1)|\{\hat{U}_t\}]]\geq(1-\bar{\epsilon}\left(s^{*}\right))(V^*(\mathbf{x}_0)-\alpha)+\\
&+\bar{\epsilon}\left(s^{*}\right)(V^*(\mathbf{x}_0)-\iota)=V^*(\mathbf{x}_0)-\alpha -\bar{\epsilon}\left(s^{*}\right)(\iota-\alpha).
\end{aligned}
\end{equation}
Take the second term in \eqref{eq:the2}, and $\forall \mathbf{x}_1\in\mathcal{X}$ given actions $\{\hat{U}_t\}$ it holds:
\begin{equation*}\begin{aligned}
&E[\hat{V}(\mathbf{x}_1)-V^*(\mathbf{x}_1)|\{\hat{U}_t\}]\geq E[\hat{r}_1+\gamma V^*(\mathbf{x}_2)+\gamma(\hat{V}(\mathbf{x}_2))-V^*(\mathbf{x}_2))-\\
&-V^*(\mathbf{x}_1)|\{\hat{U}_t\}]\geq -\alpha -\bar{\epsilon}\left(s^{*}\right)(\iota-\alpha)+\gamma E[\hat{V}(\mathbf{x}_2)-V^*(\mathbf{x}_2)].
\end{aligned}
\end{equation*}
Therefore, we can write
\begin{equation}\label{eq:the22}\begin{aligned}
\gamma& E[E_{\mathbf{x}_0}[\hat{V}(\mathbf{x}_1)-V^*(\mathbf{x}_1)|\{\hat{U}_t\}]]=\gamma E_{\mathbf{x}_0}[E[\hat{V}(\mathbf{x}_1)-V^*(\mathbf{x}_1)|\{\hat{U}_t\}]]\geq \\
\geq \gamma& \left( -\alpha -\bar{\epsilon}\left(s^{*}\right)(\iota-\alpha)+\gamma E_{\mathbf{x}_0}[E[\hat{V}(\mathbf{x}_2)-V^*(\mathbf{x}_2)]] \right)
\end{aligned}
\end{equation}
At last, substituting \eqref{eq:the21} and \eqref{eq:the22} in \eqref{eq:the2}:
\begin{equation}\label{eq:the23}\begin{aligned}
E_{\mathbf{x}_0}&[\sum_{t=0}^{\infty}\gamma^{t}\hat{r}_t]\geq V^*(\mathbf{x}_0)-(\alpha+\bar{\epsilon}\left(s^{*}\right)(\iota-\alpha))\frac{\gamma}{1-\gamma}.
\end{aligned}
\end{equation}
\qed
\end{proof}
We can interpret the results of Corollary \ref{cor:1} in the following way.
When using the exact function $\Gamma_\alpha$, the sequence of actions produced ensures that, at all times, an action is picked such that the expected sum of rewards is always larger than some bound close to the optimal. When using the approximated $\hat{\Gamma}_\alpha^{\theta^*}$, however, we obtain from the scenario approach a maximum probability of a real point not satisfying the design condition: $\|\mathbf{x}-\mathbf{x}'\|\leq \hat{\Gamma}_\alpha^{\theta^*}-\kappa^* \wedge Q^*(\mathbf{x}',\Pi^*(\mathbf{x}))<  V^*(\mathbf{x}')-\alpha$. When this happens during the execution of the c-MMDP policies it means that the agents are using delayed state information for which they do not have guarantees of performance, and the one-step-ahead value function can deviate by the worst sub-optimality gap $\iota$.
\section{Experiments}\label{sec:exp}
We set out now to quantify experimentally the impact of Algorithm \ref{al:2} on the performance and communication requirements of a benchmark c-MMDP system. First of all, it is worth mentioning that the comparison of the proposed algorithm with existing work is not possible since, to the best of our knowledge, no previous work has dealt with the problem of reducing communication when executing learned policies in a c-MMDP system. For this reason, the results are presented such that the performance of different scenarios (in terms of different $\Gamma_\alpha$ functions) is compared with the performance of an optimal policy with continuous communication.
\subsection{Benchmark: Collaborative Particle-Tag}
We evaluate the proposed solution in a typical particle tag problem (or predator-prey) \cite{lowe2017multi}. We consider a simple form of the problem with 2 predators and 1 prey. The environment is a $10\times 10$ arena with discrete states, and the predators have the actions $\mathcal{U}_i =\{\text{up},\text{down},\text{left},\text{right},\text{wait}\}$ available at each time step and can only move one position at a time. The environment has no obstacles, and the prey can move to any of the 8 adjacent states after each time step. The predators get a reward of 1 when, being in adjacent states to the prey, \emph{both} choose to move into the prey's position (tagging the prey). They get a reward of $-1$ when they move into the same position (colliding), and a reward of 0 in all other situations. A representation of the environment is presented in Figure \ref{fig:1}.
\begin{figure}
     \centering
         \includegraphics[width=0.35\linewidth]{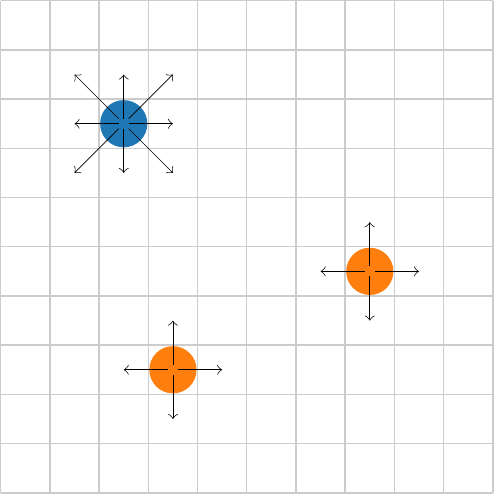}
  \caption{Particle Tag, Predators in orange, Prey in blue.}
        \label{fig:1}
\end{figure}
The global state is then a vector $\mathbf{x}_t \in \{0,1,2,...,9\}^6$, concatenating the $x,y$ position of both predators and prey. For the communication problem, we assume each agent is only able to measure its own position and the prey's. Therefore, in order to use a joint state based policy $\pi_i:\{0,1,2,...,9\}^6\to\mathcal{U}$, at each time-step predators are required to send its own position measurement to each other.
\subsection{Computation of Robustness Surrogates}
With the described framework, we first compute the optimal $Q^*$ function using a fully cooperative vanilla $Q$-learning algorithm \cite{busoniu2008comprehensive}, by considering the joint state and action space, such that $Q:\{0,1,2,...,9\}^6\times \mathcal{U}^2\to\mathbb{R}$. The function was computed using $\gamma=0.97$. We then take the joint optimal policy as $\Pi^*(\mathbf{x})=\operatorname{argmax}_{U}Q^*(\mathbf{x},U)$, and load in each predator the corresponding projection $\pi^*_i$. To evaluate the trade-off between expected rewards and communication events, we compute the function $\hat{\Gamma}_\alpha$ by solving an SVR problem as described in \eqref{eq:optimization} for different values of sensitivity $\alpha$. Then, the triggering condition for agents to communicate their measurements is $\|\mathbf{x}_t({i})-\hat{\mathbf{x}}_{t-1}(i)\|_{\infty}> \hat{\Gamma}_\alpha^{\theta^*}(\mathbf{x})-\kappa^{*}$.

The hyper-parameters for the learning of the SVR models are picked through heuristics, using a sample of size $S=10^4$ to obtain reasonable values of mis-predicted samples $s^*$ and regression mean-squared error scores. Note that $S=\frac{1}{100}|\mathcal{X}|$. To estimate the values $\overline{\epsilon}(s^*)$, a coefficient of $\beta = 10^{-3}$ was taken, and the values were computed using the code in \cite{garatti2019risk}. For more details on the computation of $\rho$-SVR models (or $\mu-$SVR)\cite{scholkopf1998shrinking} see the project code\footnote{https://github.com/danieljarne/Event-Driven-MARL}.

Figure \ref{fig:embed} shows a representation of the obtained SVR models for different values of $\alpha$, plotted over a 2D embedding of a subset of state points using a t-SNE \cite{van2008visualizing} embedding. It can be seen how for larger $\alpha$ values, more ``robust" regions appear, and with higher values of $\Gamma_\alpha$. This illustrates how, when increasing the sensitivity, the obtained approximated $\hat{\Gamma}_{\alpha}^{\theta^*}$ take higher values almost everywhere in the state space, and form clusters of highly robust points. 
\begin{figure}
     \centering
         \includegraphics[width=0.7\linewidth]{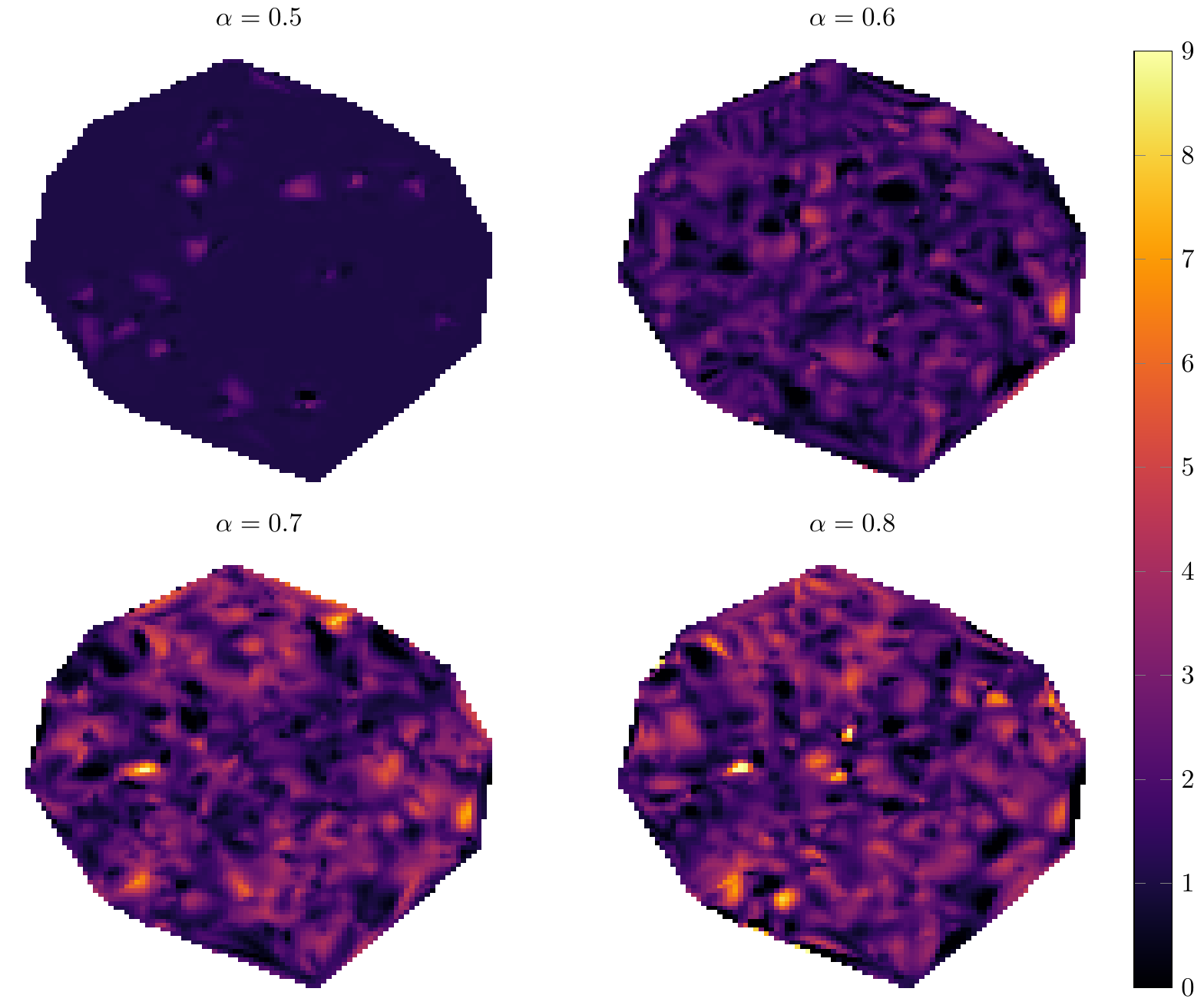}
  \caption{Obtained $\hat{\Gamma}_{\alpha}^{\theta^*}$ models on a 2D embedding.}
        \label{fig:embed}
\end{figure}
\subsection{Results}
The results are presented in Table \ref{tab:results}. We simulated in each case $1000$ independent runs of $100$ particle tag games, and computed the cumulative reward, number of communication events and average length of games. For the experiments, $\hat{E}[\cdot]$ is the expected value approximation (mean) of the cumulative reward over the $1000$ trajectories, and in every entry we indicate the standard deviation of the samples (after $\pm$). We use $\mathcal{T}_\alpha$ as the generated trajectories (games) for a corresponding parameter $\alpha$, $h(\mathcal{T}_{\alpha})$ as the total sum of communication events per game for a collection of games, and $\overline{g}:=\sum_{g\in\mathcal{T}_{\alpha}}\frac{|g|}{|\mathcal{T}_{\alpha}|}$ as the average length of a game measured over the collected $\mathcal{T}_\alpha$. For the obtained $Q^*$ function, the worst optimality gap is computed to be $\iota = 1.57$.

\begin{table}[h]\centering
\begin{tabular}{ |c|c|c|c|c|c|c|c| }
\hline
 $\alpha$& $\hat{E}[\sum_{t=0}^{\infty}\gamma^{t}r]$ & $\overline{g}$ & $h(\mathcal{T}_{\alpha})$ & $\frac{h(\mathcal{T}_{\alpha})}{\overline{g}}$&$\overline{\epsilon}(s^*)$& $\delta$\\
\hline
$0$ & $\mathbf{2.72}\pm 0.50$ & $\mathbf{10.72}\pm 0.44$ & $\mathbf{21.49}\pm 0.89$ & $\mathbf{2.00}$ & -& -\\
  \hline
$0.4$ & $\mathbf{2.72}\pm 0.53$ & $\mathbf{10.77}\pm 0.44$ & $\mathbf{19.13}\pm 0.87$ & $\mathbf{1.78}$ & $0.079$& $16.33$\\
 \hline
$0.5$ & $\mathbf{1.62}\pm 0.92$ & $\mathbf{12.45}\pm 0.58$ & $\mathbf{16.75}\pm 0.85$ & $\mathbf{1.35}$ & $0.148$ & $22.39$\\
  \hline
$0.6$ & $\mathbf{0.99}\pm 1.09$ & $\mathbf{13.71}\pm 0.71$ & $\mathbf{14.49} \pm 0.87$ & $\mathbf{1.06}$ & $0.205$ & $26.61$\\
  \hline
 $0.7$ & $\mathbf{0.93}\pm 1.08$ & $\mathbf{13.74}\pm 0.69$ &$\mathbf{14.09}\pm 0.85$ & $\mathbf{1.03}$ & $0.117$ & $26.85$\\
  \hline
     $0.8$ & $\mathbf{0.74}\pm 1.10$ & $\mathbf{14.65}\pm 0.80$ & $\mathbf{14.11} \pm 0.87$ & $\mathbf{0.96}$& $0.075$ & $28.10$\\
  \hline
 $0.9$ & $\mathbf{0.64}\pm 1.08$ & $\mathbf{14.82}\pm 0.83$ & $\mathbf{14.33} \pm 0.88$ &$\mathbf{0.96}$& $0.097$ & $31.69$\\
  \hline
\end{tabular}
\vspace{2mm}
\caption{\label{tab:results} Simulation results}
\end{table}
Let us remark the difference between the 4th and 5th column in Table \ref{tab:results}. The metric $h(\mathcal{T}_{\alpha})$ is a direct measure of amount of messages for a given value of $\alpha$. However, note that we are simulating a fixed number of games, and the average number of steps per game increases with $\alpha$: the lack of communication causes the agents to take longer to solve the game. For this reason we add the metric $h(\mathcal{T}_{\alpha})/\overline{g}$, which is a measure of total amount of messages sent versus amount of simulation steps for a fixed $\alpha$ (i.e. total amount of steps where a message \emph{could} be sent). Broadly speaking, $h(\mathcal{T}_{\alpha})$ compares raw amount of information shared to solve a fixed amount of games, and $h(\mathcal{T}_{\alpha})/\overline{g}$ compares amount of messages per time-step (information transmission rate). Note at last that there are two collaborative players in the game, therefore a continuous communication scheme would yield $h(\mathcal{T}_{\alpha})/\overline{g}=2$.

From the experimental results we can get an qualitative image of the trade-off between communication and performance. Larger $\alpha$ values yield a decrease in expected cumulative reward, and a decrease in state measurements shared between agents. Note finally that in the given c-MMDP problem, the minimum reward every time step is $\min r(\mathbf{x}_t,U,\mathbf{x}_{t+1})=-1$, therefore a lower bound for the cumulative reward is $E[\sum_{t=0}^{\infty}\gamma^{t}r]\geq -1\frac{\gamma}{1-\gamma}=-32.333$. Then, the performance (even for the case with $\alpha=0.9$ remains relatively close to the optimum computed with continuous communication.

At last, let us comment on the general trend observed regarding the values of $\alpha$. Recall the bound obtained in Corollary \ref{cor:1}, and observe that for $\alpha=0.5\Rightarrow \delta=22.39$. On average, $\hat{E}[V^*(\mathbf{x}_0)]\approx 2.72$ when initialising $\mathbf{x}_0$ at random (as seen on Table \ref{tab:results}). This yields a quite conservative bound of $\hat{E}[V^*(\mathbf{x}_0)]-\delta = -19.67$ on the expected sum of rewards, while the communication events are reduced by around $22\%$ due to the conservative computation of $\Gamma_{\alpha}$. One first source of conservativeness is in Algorithm \ref{al:2}. When computing the exact value $\Gamma_{\alpha}(\mathbf{x})=d$, it requires every point $\mathbf{x}':\|\mathbf{x}-\mathbf{x}'\|_{\infty}\leq d$ to satisfy the condition in Definition \ref{def:gam}. The number of states to be checked grows exponentially with $d$, and many of those states may not even be reachable from $\mathbf{x}$ by following the MDP transitions. Therefore we are effectively introducing conservativeness in cases where probably, for many points $\mathbf{x}$, we could obtain much larger values $\Gamma_{\alpha}(\mathbf{x})$ if we could check the transitions in the MDP. Another source of conservativeness comes from the SVR learning process and in particular, the values of $\kappa^*$. Since the states are discretised, $\|\mathbf{x}_t({i})-\hat{\mathbf{x}}_{t-1}(i)\|_{\infty}\in \{0,1,2,3,...,10\}$. Therefore, the triggering condition is effectively constrained to $\|\mathbf{x}_t({i})-\hat{\mathbf{x}}_{t-1}(i)\|_{\infty}> \lfloor \hat{\Gamma}_\alpha^{\theta^*}(\mathbf{x})-\kappa^{*}\rfloor$, which makes it very prone to under-estimate even further the true values of $\Gamma_{\alpha}(\mathbf{x})$. Additionally, for most SVR models we obtained predictions $\hat{\Gamma}_\alpha^{\theta^*}(\mathbf{x})-\kappa^{*}$ that are extremely close to the real value, so small deviations in $\kappa^*$ can have a significant impact in the number of communications that are triggered ``unnecessarily''.
\section{Discussion}
We have presented an approach to reduce the communication required in a collaborative reinforcement learning system when executing optimal policies in real time, while guaranteeing the discounted sum of rewards to stay within some bounds that can be adjusted through the parameters $\alpha$ and $\overline{\epsilon}(s^*)$ (this last one indirectly controlled by the learning of data driven approximations $\hat{\Gamma}_\alpha^{\theta^*}$). The guarantees were first derived for the case where we have access to \emph{exact} robustness surrogates $\Gamma_\alpha$, and extended to allow for surrogate functions learned through a \emph{scenario approach} based SVR optimization. In the proposed experiments for a 2-player particle tag game the total communication was reduced between $10\%-44\%$ and the communication rate by $12\%-52\%$, while keeping the expected reward sum $\hat{E}[\sum_{t=0}^{\infty}\gamma^{t}r]\in [0.68,2.76]$.

The computation of the values $\Gamma_{\alpha}(\mathbf{x})$ and the learning of the SVR models for $\hat{\Gamma}_\alpha^{\theta^*}(\mathbf{x})$ introduced significant conservativeness with respect to the theoretical bounds. A possible improvement for future work could be to compute the true values $\Gamma_{\alpha}(\mathbf{x})$ through a Monte-Carlo based approach by sampling MDP trajectories. This would yield a much more accurate representation of how ``far'' agents can deviate without communicating, and the guarantees could be modified to include the possibility that the values $\Gamma_{\alpha}(\mathbf{x})$ are correct up to a certain probability. Another option would be to compute $\Gamma_{\alpha}(\mathbf{x})$ using a different topology through embedding $\mathbf{x}$ in some higher dimensional space. At last, we can come back now to the statements in Remark \ref{rem:1}. It is now evident how having a certain physical structure in the MDP (\emph{i.e.} transition probabilities being larger for states closer in space) would help mitigate the conservativeness. An MDP with large transition jumps with respect to the sup-norm will result in more conservative and less meaningful robustness surrogates.

Other problems that branch out of this work are the implications of learning such robustness surrogate functions. These functions could be used to modify the agent policies, to sacrifice performance in favour of robustness versus communication faults or attacks. Finally, it would be insightful to compare the approaches presented in this work with ideas in the line of \cite{foerster2016learning}, where we can incorporate the communication as a binary action (to communicate or not) into the $Q$ learning algorithm, to optimise simultaneously with the sum of rewards.
\section*{Acknowledgements}
The authors want to thank Gabriel Gleizer, Giannis Delimpaltadakis and Andrea Peruffo for the useful and insightful discussions related to this work.
\bibliographystyle{ACM-Reference-Format}
\bibliography{ET_MARL.bib}
\clearpage
\appendix
\section{Appendix}
\subsection{Data driven computation of Robustness Surrogates}\label{apx:data}
The experimental results in Section \ref{sec:exp} include the computation of $\rho$-SVR models (or $\mu-$SVR)\cite{scholkopf1998shrinking}, and take advantage of the \emph{probably approximately correct} guarantees derived from the scenario approach optimization \cite{campi2020scenario} applied to adjustable size SVR models. Each value of $\alpha$ in Table \ref{tab:results} entails the solution of a scenario approach optimization using a random sample of points $\mathcal{X}_S=\{(\mathbf{x},\Gamma_{\alpha} (\mathbf{x}) )\}$ of size $S=10^{5}$. For each SVR model the hyper-parameters were tuned to obtain comparable results between them, considering that the target values $\Gamma_{\alpha}$ produced quite different results depending on the considered $\alpha$. In Table \ref{tab:svr} we collect some more details on the computed SVR models used to approximate the function $\Gamma_{\alpha}$. The $R^2$ score is computed over the training set $\mathcal{X}_S$. All the SVR models were computed using \emph{scikit-learn}'s $\mu$-SVR implementation \cite{scikit-learn}.
\begin{table}[h]\centering
\begin{tabular}{ |c|c|c|c|c|c|c|c| }
\hline
$\alpha$& $\rho$&$\tau$ &$\kappa^*$ & $R^2$ &$\frac{s^*}{S}$\\
\hline
$0.4$ & $0.01$ & $100$ & $0.00067$ & $0.999$ & $0.067$\\
\hline
$0.5$ & $0.01$ & $100$ & $0.0510$ & $0.907$ & $0.132$ \\
 \hline
$0.6$ & $0.1$ & $100$ & $0.0350$ & $0.995$ & $0.187$\\
 \hline
$0.7$ & $0.1$ & $100$ &$0.0340$ & $0.998$ & $0.103$\\
 \hline
    $0.8$ & $0.1$ & $100$ & $0.0131$ & $0.998$& $0.063$ \\
 \hline
$0.9$ & $0.1$ & $100$ & $0.01485$ &$0.999$& $0.084$ \\
 \hline
\end{tabular}
\vspace{2mm}
\caption{\label{tab:svr} SVR hyper-parameters}
\end{table}
Please note that the \emph{scikit-learn} implementation uses different nomenclature for the parameters compared to \cite{campi2020scenario} (the equivalence is $\rho = \mu$, $\tau = C$, $\kappa =\epsilon$). All the SVR models were computed with a radial basis function kernel due to the fact that, intuitively, similar state vectors would yield similar robustness values. The obtained prediction scores $R^2$ are relatively high, which could indicate some over-fitting in the models. However, this does not seem to be a problem in the implementation since the sample of training points is large enough compared to the complexity of the function $\Gamma_\alpha$, and the fact that the approximations are very conservative. 
\subsection{Experimental Results Extension}
We present in Figure \ref{fig:rewscoms} the detailed experimental results obtained in Section \ref{sec:exp}.
\begin{figure}
    \centering
        \includegraphics[width=0.49\linewidth]{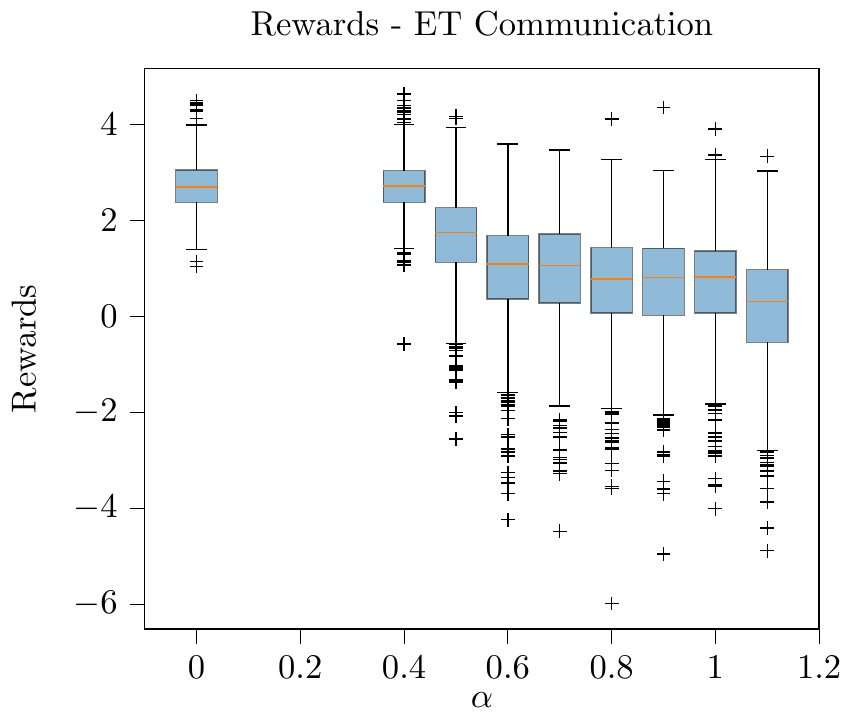}
        \includegraphics[width=0.49\linewidth]{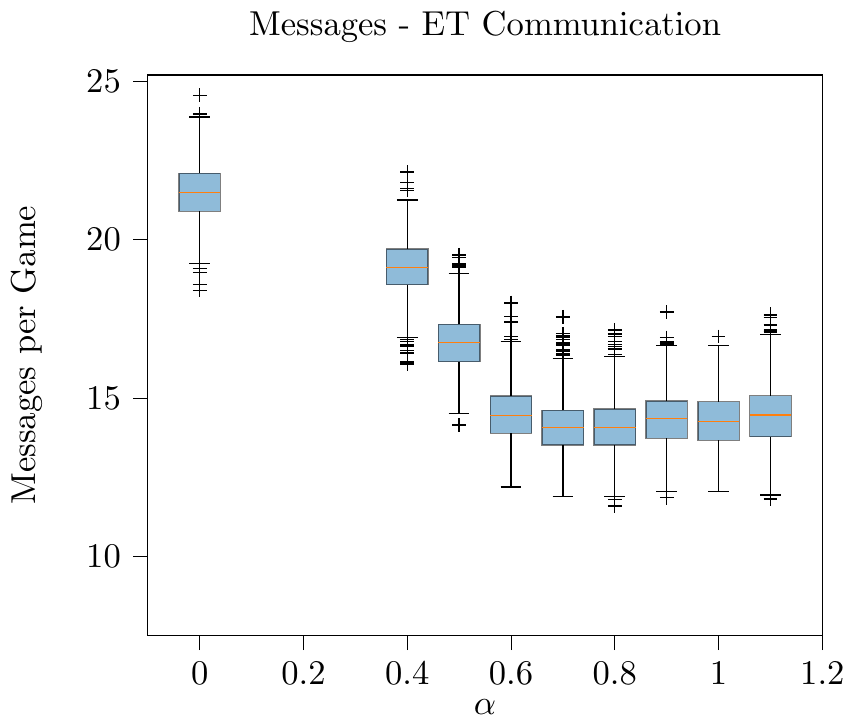}
          \caption{Cumulative rewards and Messages per game in particle tag.}
       \label{fig:rewscoms}
\end{figure}
It can be seen in Figure \ref{fig:rewscoms} how, as we increase the threshold $\alpha$, the variance in the rewards obtained increases severely. We must remark that the results presented are computed for 1000 independent games, initialised \emph{at random}. Therefore, the results in Figure \ref{fig:rewscoms} point out that the proposed event-triggered communication strategy affects the results for some initial states much more than others. In the case of the amount of messages per game, the variance does is not that heavily affected by the values of $\alpha$, and we can observe how the amount of messages plateaus at a reduction of $\approx 40\%$. This is due to the fact (pointed out in Section \ref{sec:exp}) that even though the amount of messages per time-step keeps dropping (see Table \ref{tab:results}), the amount of steps necessary to solve a particle tag game increases as the policies become less accurate.
\end{document}